\documentclass{llncs}
\usepackage{amssymb,latexsym}
\usepackage{amsmath}
\usepackage{amsfonts}
\usepackage{url}
\DeclareMathAlphabet{\mathpzc}{OT1}{pzc}{m}{it}

\newcommand{\comment}[1]{}
\renewcommand{\Pr}{\mathop{\bf Pr}\nolimits}
\newcommand{\E}{\mathop{\bf E}\nolimits}

\begin{document}

\title{More Analysis of Double Hashing \\ for Balanced Allocations}
\author{Michael Mitzenmacher\thanks{Supported in part by NSF grants CCF-1320231 and CNS-1228598.}}
\institute{Harvard University, School of Engineering and Applied Sciences\\
{\tt michaelm@eecs.harvard.edu}}
\date{}

\maketitle

\vspace{-0.3 in}

\begin{abstract}
With double hashing, for a key $x$, one generates two hash values
$f(x)$ and $g(x)$, and then uses combinations $(f(x) +i g(x)) \bmod n$
for $i=0,1,2,\ldots$ to generate multiple hash values in the range
$[0,n-1]$ from the initial two.  For balanced allocations,
keys are hashed into a hash table where each 
bucket can hold multiple keys, and each key is placed in the least
loaded of $d$ choices.  It has been shown previously that
asymptotically the performance of double hashing and fully random
hashing is the same in the balanced allocation paradigm using fluid
limit methods.  Here we extend a coupling argument used by Lueker and
Molodowitch to show that double hashing and ideal uniform hashing are
asymptotically equivalent in the setting of open address hash tables
to the balanced allocation setting, providing further insight into
this phenomenon.  We also discuss the potential for and bottlenecks
limiting the use this approach for other multiple choice hashing
schemes.
\end{abstract}

\vspace{-0.4 in}
\section{Introduction}
\vspace{-0.1 in}
\label{sec:introduction}

An interesting result from the hashing literature shows that, for open
addressing, double hashing has the same asymptotic performance as
uniform hashing.  We explain the result in more detail.  In open
addressing, we have a hash table with $n$ {\em cells} into which we
insert $m$ keys; we use $\alpha = m/n$ to refer to the load factor.
Each element is placed according to a probe sequence, which is a
permutation of the cells.  To place a key, we run through its probe
sequence in order, and place the key in the first empty cell found.
(Each cell can hold one key.)  The term uniform hashing is used to
refer to the idealized situation where the probe sequences are
independent, uniform permutations.  A key metric for such a scheme is
the search time for an unsuccessful search, which is the number of
probes until an empty cell is found.  When the load is $\alpha$, the
expected number of probes is easily shown to be $(n+1)/(n-\alpha n+1)
= (1-\alpha)^{-1} + O(1/n)$.

In contrast to uniform hashing, with double hashing, for a key $x$ one
generates two hash values $f(x)$ and $g(x)$, and then uses
combinations $(f(x) +i g(x)) \bmod n$ for $i=0,1,2,\ldots$ to generate
the permutation on $[0,n-1]$.  Here we assume that $n$ is prime, the
hash $f(x)$ is uniform over $[0,n-1]$, and $g(x)$ is uniform
$[1,n-1]$.  It might appear that limiting the space of random choices
with double hashing might significantly impact performance, but this
is not the case; it has been shown that the search time for an
unsuccessful search remains $(1-\alpha)^{-1} + o(1)$  \cite{bradford2007probabilistic,guibas1978analysis,lueker1993more}.

It is natural to ask whether similar results can be prove for other
standard hashing schemes.  For Bloom filters, Kirsch and Mitzenmacher
\cite{kirsch:bbb}, starting from the empirical analysis by Dillinger
and Manolios \cite{dillinger3312bfp}, prove that using double hashing
has asymptotically negligible effects on Bloom filter performance.
(Indeed, several publicly available implementations of Bloom filters
now use double hashing.)  Bachrach and Porat use double hashing in a
variant of min-wise independent sketches \cite{bachrach2010fast}.
Mitzenmacher and Thaler show suggestive preliminary results for double
hashing for peeling algorithms and cuckoo hashing \cite{MT}.  Leconte
consideres double hashing in the context of the load threshold for
cuckoo hashing, and shows that the thresholds are the same if one
allows double hashing to fail to place $o(n)$ keys \cite{Leconte}.
Recently, Mitzenmacher has shown that double hashing asymptotically
has no effect on the load distribution in the setting of balanced
allocations \cite{mitzenmacher2014balanced}; we describe this result
further in the related work below.

As a brief review, the standard balanced allocation paradigm works as follows: suppose
$m$ balls (the keys) are sequentially placed into $n$ bins (hash table buckets), where each ball is
placed in the least loaded of $d$ uniform independent choices of the
bins.  Typically we think of each of these $d$ choices as being obtained from a
random hash function; we therefore refer to this setting as using random hashing.
We use the standard balls and bins nomenclature for this setting (although one could
correspondingly use keys and buckets.)  
In the case where the number of balls and bins are equal,
that is $m = n$, the maximum load (that is, the maximum number of balls in
a bin) is $\frac{\log \log n}{\log d} + O(1)$,
much lower than the $\frac{\log n}{\log \log n} (1 +o(1))$ obtained
where each ball is placed according to a single uniform choice \cite{ABKU}.
Further, using a fluid limit model that yields a family of differential 
equations describing the balanced allocations process, one can determine,
for any constant $j$, the asymptotic fraction of bins of load $j$ as
$n$ goes to infinity, and Chernoff-type bounds hold that can bound the fraction
of bins of load $j$ for finite $n$ \cite{MitzThesis}.  
(These results extend naturally when $m = cn$ for a constant $c$;
the maximum load remains $\frac{\log \log n}{\log d} + O(1)$.)  

For balanced allocations in conjunction with double hashing, the $j$th
ball obtains two hash values, $f(j) \in [0,n-1]$ and $g(j) \in
[1,n-1]$, chosen uniformly from these ranges.  The $d$ choices for the
$j$th ball are then given by $h(j,k) = (f(j) + k g(j)) \bmod n$,
$k=0,1,\ldots,d-1$, and the ball is placed in the least loaded.  For
convenience in this paper we take $n$ to be prime, but the results can
be modified straighforwardly by having $g(j)$ chosen relatively prime
to $n$.  In particular, if $m$ is a power of 2, as is natural in practice, by
having $g(j)$ uniformly chosen from the odd numbers in $[1,n-1]$ we
obtain analogous results.  

The purpose of this paper is to provide an alternative proof that
double hashing has asymptotically negligible effects in the setting of
balanced allocations.  Specifically, we extend a coupling argument
used by Lueker and Molodowitch to show that double hashing and ideal
uniform hashing are asymptotically equivalent in the setting of open
address hash tables to the balanced allocation setting.  We refer to
their argument henceforth as the LM argument.  As far as we
are aware, this is the first time this coupling approach has been
used for a hashing scheme outside of open addressing.  Adapting the LM argument
gives new insights into double hashing for balanced allocations,
as well as to the potential for this approach
to be used for other multiple choice hashing schemes.  

In particular, our modification of the LM argument involves
significant changes.  For reasons we explain, the LM argument does not
seem to allow a direct coupling with random hashing; instead, we 
couple with an intermediary process, which is equivalent to random
hashing plus some small bias that makes bins with heavy load slightly
more likely.  We then argue that this added bias does not
affect the asymptotic performance of the balanced allocations process,
providing the desired connection between the balanced allocation
process with random hashing and double hashing.  Specifically, for
constant $d$, the maximum load remains $\frac{\log \log n}{\log d} +
O(1)$ with high probability, and the asymptotic fraction of bins with
constant load $j$ can be determined using the method of differential
equations.

\vspace{-0.2 in}
\subsection{Related Work}
\vspace{-0.1 in}

The balanced allocations paradigm, or the power of two choices, has
been the subject of a great deal of work.  See, for example, the
survey articles \cite{KSurvey,TwoSurvey} for references and
applications. 

The motivation for this paper stems from recent work showing that
the asymptotic fraction of bins of each load $j$ (for constant $j$)
for double hashing can be determined using the same differential
equations describing the behavior for random hashing
\cite{mitzenmacher2014balanced}.  
Using insight from this approach
also provides a proof that using double hashing, for a constant number
of choices $d$, the maximum load is $\log \log n / \log d + O(1)$ with
high probability using double hashing.  The latter result is obtained
by modifying the layered induction approach of \cite{ABKU} for random
hashing.  Here we provide an alternative way of obtaining these
results by a direct coupling with a slightly modified version of
random hashing, based on the LM argument.  The paper 
\cite{mitzenmacher2014balanced} also contains discussion of related work.    

Of course, our work is also highly motivated by the chain of work
\cite{bradford2007probabilistic,guibas1978analysis,lueker1993more,SSCH}
regarding the classical question of the behavior of double hashing for
open address hash tables, where empirical work had shown that the
difference in performance, in terms of the average length of an
unsuccessful search sequence, appeared negligible.  Theoretically, the
main result showed that for a table with $n$ cells and $\alpha n$ keys
for a constant $\alpha$, the number of probed locations in an
unsuccessful search was (up to lower order terms) $1/(1-\alpha)$  for
both double hashing and uniform hashing \cite{lueker1993more}.  We
have not seen this methodology applied to other hashing schemes such
as balanced allocations, although of course the issue of limited
randomness is pervasive; a recent example include studying the use of $k$-wise
independent hash functions for linear probing for small constant $k$ \cite{ppr,pat}.

\vspace{-0.2 in}
\section{Coupling Double Hashing and Random Hashing}
\vspace{-0.1 in}
\label{sec:background}

Before delving into our proof, it is worth describing the LM argument
at a high level, as well as changes needed in the balanced allocation context.

Consider the setting of open address hashing, where $m'$ keys have
been placed into a table of size $n$ using uniform hashing.  Suppose
now we consider placing the next key using double hashing instead of
random hashing.  The LM argument shows that we can couple the
decisions so that, with high probability (by which we mean $1-o(1)$),
the end result in terms of where the key lands is the same.
Inductively, this means that if we start from an empty table, we can
couple the two processes step by step, and as long as the coupling
holds, the two tables will appear exactly the same.

However, there is a problem.  Let us suppose that we run the processes
for $m = \alpha n$ keys.  While the two processes match up on any
single step with high probability, this probability is not high enough
(it is $\Omega(1/n)$) to guarantee that the two processes couple over all $m$ insertions of
balls with high probability.  At some point, the two processes will
very likely deviate, and we need to consider that deviation.

In fact, the LM argument enforces that the deviation occur in a
particular way.  They show that the probability a key ends in any
given position from double hashing is at most only $1+\delta$ times
the probability a key ends in any given position from uniform hashing
for a $\delta$ that is $o(1)$.  The coupling then places the key
according to double hashing with probability $1/(1+\delta)$ in both
tables, and with probability $\delta/(1+\delta)$ it places the key to
yield the appropriate distribution from uniform hashing.  As a result,
both tables follow the placement given by uniform hashing;
hence, in the rare case where coupling fails, it
fails in such a way that the double hashing process has obtained a key
placed according to uniform hashing.

When such a failure occurs, to the double hashing process, the key
appears as a randomly placed {\em extra} key that has entered the
system and that was not placed according to double hashing.  The LM argument then
makes uses of the following property: adding such an extra key only
makes things worse, in that at the end of the double hashing process
every hash cell occupied by a key if the extra key hadn't been added
will still be occupied.  This is a form of {\em domination} that the
LM argument requires.

As $\delta = o(1)$, the LM argument concludes by showing that if
we run the coupled process for $\alpha m + o(m)$ keys for a suitably
chosen $o(m)$, then at least $\alpha m$ keys will be added in the double
hashing process according to double hashing.  That is, 
the number of extra keys added is asymptotically negligible, giving the
desired result: double hashing is stochastically dominated by uniform
hashing with an asymptotically negligible number of extra keys, which does not affect the
high order $1/(1-\alpha)$ term for an unsuccessful search.  

We attempt to make an analogous argument in the double hashing setting
for balanced allocations.  A problem arises in that it seems we cannot
arrange for the coupling to satisfy the requirements of the original
LM argument.  As mentioned, in the open address setting, each position is only at
most $1+\delta$ times as likely to obtain a key (with high probability
over the results of the previous steps).  This fact is derived from Chernoff
bounds that hold because each cell has a reasonable chance of being
chosen; when there are $m'$ cells filled, each cell is the next filled
with probability approximately $1/(n-m')$.  But this need not be the case
in the balanced allocation setting.  As an example, consider the $d$th
most loaded bin; suppose for convenience it is the only bin with a
given load.  Using random hashing, the probability it receives a ball
is $O(n^{-d})$, as all $d$ choices have to be among the $d$ most
loaded bins.  Using double hashing, the probability it receives a ball
could be $\Omega(n^{-2})$, if the $d$ most loaded bins are in an arithmetic
progression that align with the double hashing.  While in this example
the probability the $d$ choices align this way is rare, in general the
probability that some bin is significantly more likely to obtain a
ball when using double hashing does not appear readily swept into
$o(1)$ failure probabilities.

However, intuitively, by concentration, this problem can only occur
for bins that are rarely chosen under uniform hashing; that is, for
bins with high load.  We therefore can resolve the issue by not
coupling double hashing with random hashing directly, but instead
slightly perturbing the distribution from random hashing to give
slightly more weight to heavily loaded bins, enough to cope with the
relative looseness in the concentration bounds for rare events.  We
then show that this small modification to random hashing does not
affect the characteristics of the final distribution of balls into
bins that we have described above.

\vspace{-0.2 in}
\section{Modified Random Hashing}
\label{sec:modrh}
\vspace{-0.1 in}

We start by defining the {\em modified random hashing} process that we couple with.
Given a balls and bins configuration, we do the following to place the next
ball:
\begin{itemize}
\item   with probability $n^{-0.4}$, we place the ball uniformly at random;
\item   with all remaining probability, we place the ball according to the least loaded of $d$ choices (with ties broken randomly).
\end{itemize}

We briefly note the following results regarding this modified random hashing.

\begin{lemma}
Let $i$, $d$, and $T$ be constants.  Suppose $m=Tn$ balls are
sequentially thrown into $n$ bins according to the modified random hashing process.  Let $X_i(T)$ be the number of
bins of load at least $i$ after the balls are thrown.  Let $x_i(t)$
be determined by the family of differential equations 
$$\frac{dx_i}{dt} = x_{i-1}^d - x_{i}^d,$$
where $x_0(t) = 1$ for all time and $x_i(0) = 0$ for $i \geq 1$.  
Then with probability $1-o(1)$, 
$$\frac{X_i(T)}{n} = x_i(T) + o(1).$$
\end{lemma}

\begin{lemma}
\label{lem:mrh2}
Let $d$ and $T$ be constants.  Suppose $m=Tn$ balls are
sequentially thrown into $m$ bins according to the modified random hashing process.  
Then the maximum load is $\frac{\log \log n}{\log d} + O(1)$, where the $O(1)$ term
depends on $T$.  
\end{lemma}

Both proofs follow readily from the known proofs of these statements under random hashing,
with small changes to account for the modification.  Intuitively, only 
$mn^{-0.4} = Tn^{0.6}$ balls are distributed randomly, which with high
probability affects $o(n)$ bins by at most $O(1)$ amounts.  Hence, one would not
expect the modification to the random hashing process to change the load distribution
substantially.  More details are given in the Appendix.   

\vspace{-0.2 in}
\section{The Coupling Proof}
\vspace{-0.1 in}

We now formalize the coupling proof.  While 
we generally follow the description of Lueker and Molodowitch, our different setting naturally requires
changes and some different terminology.  

Recall that we assume that there is a table of $n$ bins, where $n$ is
a prime.  We may use the term {\em hash pair} to refer to one of the
$n(n-1)$ possible pairs of hash values $(f(j),g(j))$.  We aim to
consider the outcomes when $m = cn$ balls are placed using double
hashing.  We refer to the bin state as the ordered list
$(b_1,b_2,\ldots,b_n)$, where $b_i$ is the number of balls in the
$i$th bin.  For any bin state, for every bin $z$, let $\hat{\eta}(z)$
be the number of hash pairs that would cause $z$ to obtain the next
ball. It follows that the probability ${\eta}(z)$ that $z$ is the next
bin to obtain a ball is $\hat{\eta}(z)/(n(n-1))$.

When considering modified random hashing with $d$ choices, we assume the $d$
choices are made without replacement.  This choice does not matter, as
it is known the difference in performance between choosing with and without replacement
is negligible.  Specifically, for constant $d$, the expected number of balls that would choose
some bin more than once is constant, and is $O(\log n)$ with high probability;  
this does not affect the asymptotic behavior of the system.  In our setting,
since with double hashing the choices are without replacement, it makes the argument details
somewhat easier.

Similarly, we technically need to consider what to do if there is a
tie for the least loaded bin.  
For convenience, in case of a tie we assume that the tie is broken randomly
among the bins that share the least load, but in the following coupled fashion.
At each step, we assume a ranking is given to the bins (according to a random
permutation of $[1,n]$);  the rankings at each step are independent and uniform.
In case of tie in the load, the rank is used the break the tie.  Note that, 
at each step, we then have a total ordering on the bins, where the order is
determined first by the load and then by the rank.  We refer to the $j$th ordered bin,
with the following meaning;  the first ordered bin is the heaviest loaded with lowest priority in tie-breaking,
and the $n$th ordered bin is the least loaded with the highest priority in tie-breaking.
Hence, with random hashing, the $j$th ordered bin obtains the next ball with probability
$$\frac{d}{n} \frac{{{j-1} \choose {d-1}}}{{{n-1} \choose {d-1}}}.$$
The $d/n$ term represents that the $j$th ordered bin must be one of the $d$ choices;  the other term
represents that the remaiining $d-1$ choices must be from the top $j-1$ ordered elements.

We extend the domination concept used in the LM argument in the natural way.
We say a bin state $B = (b_1,b_2,\ldots,b_n)$ dominates a bin state $A
= (a_1,a_2,\ldots,a_n)$ if $b_i \geq a_i$ for all $i$.  We may write
$B \succeq A$ when $B$ dominates $A$.  The following is the key point regarding
domination:
\begin{lemma}
\label{lem:dom}
If $B \succeq A$, and we insert a ball into a table $B$ to
obtain $B'$ and the same ball into $A$ to obtain $A'$ by using the least
loaded of $d$ choices, then (whether we use double hashing, 
random hashing, or modified random hashing) $B' \succeq A'$.  
\end{lemma}
\begin{proof}
Suppose the bin choices are $i_1,i_2,\ldots,i_d$.  Without loss of
generality let $i_1$ be the least loaded of these choices in bin state $A$ (or the bin
chosen by our tie-breaking scheme). If $i_1$ is not chosen as the least loaded in $B$,
it must be because $b_{i_1} > a_{i_1}$, and hence even after the ball is placed,
$B' \succeq A'$.  {\hspace*{\fill}\rule{6pt}{6pt}\bigskip}
\end{proof}
Our goal now is to show that if we have a table which has been filled
up to that point by modified random hashing, we can couple
appropriately.  That is, we can couple by using the result of a double hashing step 
with high probability, and with some small probability we use a modified random hashing
step, giving the double hashing process an {\em extra} ball.  

To begin, we note that with modified random hashing, the $j$th ordered bin
obtains a ball with probability 
$$ p_j  = (1-n^{-0.4})\left( \frac{d}{n} \frac{{{j-1}\choose {d-1}}}{{{n-1} \choose {d-1}}}\right ) + n^{-0.4}\frac{1}{n} 
    = (1-n^{-0.4})\left( \frac{d}{n} \frac{{{j-1}\choose {d-1}}}{{{n-1} \choose {d-1}}}\right ) + n^{-1.4}
$$
In the right hand side of the first equality, the first term expresses the probability that there are $d$ choices and that the ball chooses the $j$th order bin.  The second term arises from the probability that the ball is placed randomly after choosing a single bin.

We wish to show the following:
\begin{lemma}
\label{lem:coupling}
Suppose a bin $z$ is $j$th in the ordering after
starting with an empty table and adding $n'$ balls by
modified random hashing.  Then 
$${\eta}(z) \leq p_j \left (1 + n^{-0.01} \right )$$
except with probability $ne^{-0.05n}$, where this probability is over the random
bin state obtained from the $n'$ placed balls.
\end{lemma}

We remark that the constants here were chosen for convenience and not optimized;  
this is sufficient for our asymptotic statements.  

\begin{proof}
If $z$ is $j$th in the ordering, we have
$$\E[{\eta}(z)] = \frac{d}{n} \frac{{j-1 \choose d-1}}{{n-1 \choose d-1}} \mbox{    ;    }
\E[\hat{\eta}(z)] = d(n-1) \frac{{j-1 \choose d-1}}{{n-1 \choose d-1}}.$$
Here we use the fact that, under modified random hashing, the ordering of the bins form a uniform permutation.    
Hence, in expectation, double hashing yields the same probability for a bin obtaining a ball as random hashing.
However, we must still show that individual probabilities are close to their expectations.

We first show that when $\E[{\eta}(z)]$ is sufficiently large 
then ${\eta}(z)$ is close to its expectation, which is unsurprising.
When $\E[{\eta}(z)]$ is small, so that tail bounds are weaker, 
we are rescued by our modification to $p_j$;  the 
additional $n^{-1.4}$ skew in the
distribution that we have added for modified random hashing will
our desired bound between ${\eta}(z)$ and $p_j$.  

In this case, we use martingale bounds;  we use martingales
instead of Chernoff bounds because there is dependence among the behavior of the $d(n-1)$
hash pairs that include bin $z$.  

We set up the martingale as follows.  We refer to the bins as bins 1 to $n$.  Without loss of generality let $z$ be the last bin (labeled 
$n$) and 
and let $Z_i$ be the rank in the bin ordering of the $i$th bin, for $i=1$ to $n-1$.  We expose the $Z_i$ one
at a time to establish a Doob martingale \cite[Section 12.1]{MU}.  Let 
$$Y_i = \E[\hat{\eta}(z)~|~Z_1,\ldots,Z_i].$$
Note $Y_0 = \E[\hat{\eta}(z)]$ and $Y_{n-1} = \hat{\eta}(z)$.  
We claim that 
$$|Y_i - Y_{i-1}| \leq d^2.$$
To see this, note that changing our permutation of the ordering of the bins by switching
the rank order of two bins $a$ and $b$ can only affect the hash pairs that include 
$z$ and $a$ or $z$ and $b$;  there are fewer than $d^2$ such sequences, since there are
${d \choose 2}$ hash pairs than include any pair of bins (determined by which of the $d$
hashes each of the two bins corresponds to).  

Hence we can apply the standard Azuma-Hoeffding inequality (see, e.g., \cite[Theorem 12.4]{MU})
to obtain
$$\Pr(|Y_{n-1} - Y_0| \geq \lambda) \leq 2e^{-\lambda^2/(2(n-1)d^2)}.$$
Hence 
$$\Pr(|\hat{\eta}(z) - \E[\hat{\eta}(z)]| \geq \lambda) \leq 2e^{-\lambda^2/(2(n-1)d^2)}.$$

We now break things into cases. First, suppose $z$ and $j$ are such that $\E[\hat{\eta}(z)]
\geq n^{0.55}$. We choose $\lambda = n^{0.53}$ to obtain
$$\Pr(|\hat{\eta}(z) - \E[\hat{\eta}(z)]| \geq n^{0.53}) \leq 2e^{-n^{1.06}/(2(n-1)d^2)} \leq e^{-n^{0.05}}.$$
for sufficiently large $n$.  
Hence 
$$\Pr(|{\eta}(z) - \E[{\eta}(z)]| \geq n^{-0.47}/(n-1)) \leq e^{-n^{0.05}}.$$
We also note that in this case
$$p_j  \geq {(1-n^{-0.4})}{\E[{\eta}(z)]},$$
so 
$$p_j {(1+2n^{-0.4})} \geq {\E[{\eta}(z)]}$$
for large enough $n$.  
It follows that 
$$\Pr({\eta}(z) - (1+2n^{-0.4})p_j \geq n^{-0.47}/(n-1)) \leq e^{-n^{0.05}}.$$
Further, $p_j \geq n^{0.55}(1-n^{-0.4})/n(n-1)$.  
Simpliyfing the above we find
$$\Pr({\eta}(z) - p_j \geq 2n^{-0.4}p_j + n^{-0.47}/(n-1)) \leq e^{-n^{0.05}},$$
which implies 
$$\Pr({\eta}(z) - p_j \geq n^{-0.01}p_j) \leq e^{-n^{0.05}}.$$
Hence ${\eta}(z) \leq (1+n^{-0.01})p_j$ with very high probability over the bin state.  

Now, consider when $z$ and $j$ are such that $\E[\hat{\eta}(z)]
\leq n^{0.55}$. We again choose $\lambda = n^{0.53}$ to obtain
$$\Pr(|{\eta}(z) - \E[{\eta}(z)]| \geq n^{-0.47}/(n-1)) \leq e^{-n^{0.05}}.$$
In this case, $p_j \geq n^{-1.4}$ and hence greater than $\E[{\eta}(z)]$ for sufficiently large $n$.
Hence 
$$\Pr({\eta}(z) - p_j \geq n^{-0.47}/(n-1)) \leq e^{-n^{0.05}},$$
and therfore
$$\Pr({\eta}(z) - p_j \geq n^{-0.05}p_j) \leq e^{-n^{0.05}}.$$

In both cases, we have ${\eta}(z) \leq (1+n^{-0.01})p_j$ with probability at most $e^{-n^{0.05}}$;
a union bound gives the result.  {\hspace*{\fill}\rule{6pt}{6pt}\bigskip}
\end{proof}

From this point, we can return to following the LM argument.  We have shown that the probability
a bin is chosen using double hashing is at most $(1+\delta)$ times that of modified random hashing
for $\delta = n^{-0.01}$, with high probability.
We therefore consider the following algorithm to bound the performance of throwing $m = cn$ balls
into $n$ bins using double hashing.  In what follows, we discuss a bin $z$ that we consistently
make $j$th in the ordering, so with modified random hashing the probability a ball lands in $z$
is $p_j$, and with double hashing this probability is $\eta(z)$.   
\begin{enumerate}
\item We throw $(1+2\delta)m$ balls.  
\item If, at any step, we have $\eta(z) > (1+\delta)p_j$ for any bin $z$, the algorithm fails and we stop.
Otherwise, we place balls as follows.  
\item At each step, with probability $1/(1+\delta)$, we place a ball according to double
hashing.
\item Otherwise, with probability $\delta/(1+\delta)$, we place a ball with probability $\delta^{-1} \left (  (1+\delta)p_j - \eta(z) \right)$ into bin $z$.  
\end{enumerate}

\begin{theorem}
With high probability, the algorithm above places $(1+2\delta)m$ balls according to modified random hashing,  
and at least $m$ balls are placed according to double hashing.  The final bin state therefore dominates that 
of placing $m$ balls using double hashing with high probability.  
\end{theorem}
\begin{proof}
A simple calculation shows that each ball lands in the $j$th ordered bin with probability 
$$ \frac{1}{1+\delta} \eta(z) + \frac{\delta}{1+\delta} \left ( \frac{(1+\delta)p_j}{\delta } - \frac{\eta(z)}{\delta} \right) = p_j.$$
So each ball is placed with the same distribution as for
 modified random hashing, as long as no bin has $\eta(z) > (1+\delta)p_j$.  
By Lemma~\ref{lem:coupling}, the probability of such a failure is union bounded by $mne^{-0.05n}$ over the $m$ steps
of adding balls.  

Let $B$ be the number of balls placed by double hashing when using the above algorithm. Then
$$\E[B] = (1+2\delta)m/(1+\delta) > (1+\delta/2)m.$$
A simple Chernoff bound \cite[Exercise 4.13]{MU} gives
$$\Pr(B \leq m) \leq e^{-2m(\delta/2)^2/(1+2\delta)} \leq e^{-0.97n}$$
for $n$ sufficiently large and $m = cn$ for a constant $c$.  

The extra balls placed by modified random hashing are handled via our domination result, Lemma~\ref{lem:dom}.
{\hspace*{\fill}\rule{6pt}{6pt}\bigskip}
\end{proof}

The following corollary is immediate from the domination.
\begin{corollary}
\label{cor:mrh2}
Let $d$ and $T$ be constants.  Suppose $m=Tn$ balls are
sequentially thrown into $m$ bins according to double hashing.
Then the maximum load is $\frac{\log \log n}{\log d} + O(1)$, where the $O(1)$ term
depends on $T$.  
\end{corollary}

This next corollary follows from the fact that the algorithm shows that, step by step,
the double hashing process and the modified hashing process are governed by the same
family of differential equations, as the probability of going into a bin of a given
load differs by $o(1)$ between the two processes.
\begin{corollary}
\label{cor:mrh1}
Let $i$, $d$, and $T$ be constants.  Suppose $m=Tn$ balls are
sequentially thrown into $n$ bins according to double hashing.  Let $X_i(T)$ be the number of
bins of load at least $i$ after the balls are thrown.  Let $x_i(t)$
be determined by the family of differential equations 
$$\frac{dx_i}{dt} = x_{i-1}^d - x_{i}^d,$$
where $x_0(t) = 1$ for all time and $x_i(0) = 0$ for $i \geq 1$.  
Then with probability $1-o(1)$, 
$$\frac{X_i(T)}{n} = x_i(T) + o(1).$$
\end{corollary}

\vspace{-0.2 in}

\section{Conclusion}
We have shown that the coupling argument of Lueker and Molodowitch
can, with some modification of the standard random hashing process,
yield results for double hashing with the balanced allocations
framework.  It is worth considering if this approach could be
generalized further to handle other processes, most notably cuckoo
hashing and peeling processes, where double hashing similarly seems to
have the same performance as random hashing \cite{MT}.  The challenge
here for cuckoo hashing appears to be that the state change on entry
of a new key is not limited to a single location; while only one cell
in the hash table obtains a key, other cells become potential future
recipients of the key if it should move, effectively changing the
state of those cell.  This appears to break the coupling method of the
LM argument, which conveniently can forget the choices involved after
an item is placed.  The issue similarly arises for peeling processes,
Robin Hood hashing, and other hashing schemes involving multiple
choice.  However, we optimistically suggest there may be some way to
further modify and extend this type of argument to remove this problem.

\vspace{-0.2 in}

\bibliographystyle{plain}

\begin{thebibliography}{99}

\vspace{-0.1 in}

\bibitem{ABKU}
Y. Azar, A. Broder, A. Karlin, and E. Upfal.
\newblock Balanced allocations.
\newblock {\em SIAM Journal of Computing} 29(1):180-200, 1999.

\bibitem{bachrach2010fast} Y. Bachrach and E. Porat.
    Fast pseudo-random fingerprints.  
    {Preprint arXiv:1009.5791}, 2010.

\bibitem{bradford2007probabilistic}
P. Bradford and M. Katehakis.
A probabilistic study on combinatorial expanders and hashing.
{\em SIAM Journal on Computing}, 37(1):83-111, 2007.

\bibitem{dillinger3312bfp}
P.C. Dillinger and P.~Manolios.
\newblock {Bloom Filters in Probabilistic Verification}.
In \newblock {\em Proc. of the 5th Intl. Conference on Formal
  Methods in Computer-Aided Design}, pp. 367-381, 2004.

\bibitem{EK}
S. N. Ethier and T. G. Kurtz.
{\em Markov Processes:  Characterization and Convergence}.
John Wiley and Sons, 1986.

\bibitem{guibas1978analysis}
L. Guibas and E. Szemeredi.
The analysis of double hashing.
{\em Journal of Computer and System Sciences}, 16(2):226-274, 1978.

\bibitem{kirsch:bbb}
A.~Kirsch and M.~Mitzenmacher.
\newblock {Less hashing, same performance: Building a better Bloom filter}.
\newblock {\em Random Structures \& Algorithms}, 33(2):187-218, 2008.

\bibitem{KSurvey} A. Kirsch,  M. Mitzenmacher and G. Varghese.
    Hash-Based Techniques for High-Speed Packet Processing. 
     In {\em Algorithms for Next Generation Networks}, (G. Cormode and M. Thottan, eds.), pp. 181-218, Springer London, 2010.

\bibitem{Kurtz}
T. G. Kurtz.
Solutions of Ordinary Differential
Equations as Limits of Pure Jump {M}arkov Processes.
{\em Journal of Applied Probability},
Vol. 7, 1970, pp. 49-58.

\bibitem{Leconte}
M. Leconte.  Double hashing thresholds via local weak convergence.
In {\em Proceedings of the 51st Annual Allerton Conference on Communication, Control, and Computing},
pp. 131-137, 2013. 

\bibitem{lueker1993more}
G. Lueker and M. Molodowitch.
More analysis of double hashing.
{\em Combinatorica}, 13(1):83-96, 1993.

\bibitem{MitzThesis} M. Mitzenmacher. The power of two choices
in randomized load balancing.  Ph.D. thesis, 1996.

\bibitem{mitzenmacher2014balanced}
M. Mitzenmacher.
Balanced allocations and double hashing.
In {\em Proceedings of the 26th ACM Symposium on Parallelism in Algorithms and Architectures},
pp. 331--342, 2014.

\bibitem{TwoSurvey} M. Mitzenmacher, A. Richa, and R. Sitaraman.
    The Power of Two Choices: A Survey of Techniques and
    Results. In {\em Handbook of Randomized Computing}, (P. Pardalos, S. Rajasekaran, J. Reif, and J.
    Rolim, edds), pp. 255-312, Kluwer Academic Publishers, Norwell, MA, 2001.

\bibitem{MT}
M. Mitzenmacher and J. Thaler.
Peeling Arguments and Double Hashing.
In {\em Proceedings of the 50th Annual Allerton Conference on Communication, Control, and Computing},
pp. 1118-1125, 2012. 

\bibitem{MU}
M. Mitzenmacher and E. Upfal.
{\em Probability and computing: Randomized algorithms and probabilistic analysis},
2005, Cambridge University Press.

\bibitem{ppr}
A. Pagh, R. Pagh, and M. Ruzic.
Linear Probing with 5-wise Independence. {\em SIAM Review}, 53(3):547-558, 2011.  

\bibitem{pat}
M. Patra{\c{s}}cu and M. Thorup.
On the $k$-independence required by linear probing and minwise independence.  
In {\em Proceedings of ICALP}, pp. 715-726, 2010.

\bibitem{SSCH}
J. Schmidt and A. Siegel.
The analysis of closed hashing under limited randomness.
In {\em Proceedings of the 22nd Annual ACM Symposium on Theory of Computing},
  pp. 224-234, 1990.

\bibitem{Wormald}
N.C. Wormald.
\newblock Differential equations for random processes and random graphs.
\newblock {\em The Annals of Applied Probability}, 5(1995), pp. 1217--1235.

\end{thebibliography}

\section*{Appendix}

We briefly sketch the proofs of the following results regarding modified random hashing
that we discussed in section~\ref{sec:modrh}. 
\setcounter{lemma}{0}

\begin{lemma}
Let $i$, $d$, and $T$ be constants.  Suppose $m=Tn$ balls are
sequentially thrown into $n$ bins according to the modified random hashing process.  Let $X_i(T)$ be the number of
bins of load at least $i$ after the balls are thrown.  Let $x_i(t)$
be determined by the family of differential equations 
$$\frac{dx_i}{dt} = x_{i-1}^d - x_{i}^d,$$
where $x_0(t) = 1$ for all time and $x_i(0) = 0$ for $i \geq 1$.  
Then with probability $1-o(1)$, 
$$\frac{X_i(T)}{n} = x_i(T) + o(1).$$
\end{lemma}
\begin{proof}
(Sketch.)  We note that this result holds for either modified random hashing or random hashing.
For random hashing, the result is a well known application of the fluid limit approach.
Specifically, suppose we let $X_i(t)$ be a random variable denoting the number of 
bins with load at least $i$ after $tn$ balls have been thrown, 
and let $x_i(t) = X_i(t)/n$.  For $X_i$ to increase when
a ball is thrown, all of its choices must have load at least $i-1$, but not all of them
can have load at least $i$.  Let us first consider the case of random hashing.  For $i \geq 1$,
$$\E[X_i(t + 1/n) - X_i(t)] = (x_{i-1}(t))^d - (x_{i}(t))^d.$$
Let $\Delta(x_i) = x_i(t + 1/m) - x_i(t)$ and $\Delta(t) = 1/n$.  Then the above can be written as:
$$\E \left [ \frac{\Delta(x_i)}{\Delta(t)} \right ] = (x_{i-1}(t))^d - (x_{i}(t))^d.$$
In the limit as $m$ grows, we can view the limiting version of the above equation as
$$\frac{dx_i}{dt} = x_{i-1}^d - x_{i}^d.$$
The works of Kurtz and Wormald \cite{EK,Kurtz,Wormald} justify 
convergence of the random hashing process to the solution of the differential equations.
Specifically, it follows from Wormald's theorem \cite[Theorem 1]{Wormald} that
$$X_i(t) = nx_i(t) + o(n)$$
with probability $1-o(1)$, which matches the desired result.

Now we note that for the modified hashing process, with the corresponding variables, we have  
\begin{eqnarray*}
\E[X_i(t + 1/n) - X_i(t)] & = & (1-m^{-0.4})(x_{i-1}(t))^d - (x_{i}(t))^d + x_{i-1}(t)n^{-0.4} \\
& = & (x_{i-1}(t))^d - (x_{i}(t))^d + o(1).
\end{eqnarray*}
Wormald's theorem \cite[Theorem 1]{Wormald} allows $o(1)$ additive terms and yields the same result,
namely $X_i(t) = nx_i(t) + o(n)$ with probability $1-o(1)$.   {\hspace*{\fill}\rule{6pt}{6pt}\bigskip}
\end{proof}

\begin{lemma}
Let $d$ and $T$ be constants.  Suppose $m=Tn$ balls are
sequentially thrown into $m$ bins according to the modified random hashing process.  
Then the maximum load is $\frac{\log \log n}{\log d} + O(1)$, where the $O(1)$ term
depends on $T$.  
\end{lemma}
\begin{proof}
The proof is a simple modification of the layered induction proof of \cite[Theorems 3.2 and 3.7]{ABKU}.
For convenience, we consider just the case of $m = n$ to present the main idea, which corresponds to
\cite[Theorem 3.2]{ABKU}.  The theorem inductively shows that for balanced allocations
with random hashing the number of bins with load at least $i$
is bounded above with high probability by 
$$\beta_i = \frac{ne^{(d^{i-6}-1)/(d-1)}}{(2e)^{d^{i-6}}}$$
 for $i \geq 6$ and $i < i^*$ for some $i^* \leq \ln \ln n/\ln d + O(1)$, where for $i^*$
we have $\beta_{i^*}^d/n^d \leq 2 \ln n$.   

The same results hold with essentially the same induction when using the modified random hashing;
however, one must stop the induction earlier.  In particular, the probability that a ball lands in
a bin with load at least $i$ is now given by $(1-n^{-0.4})\beta_{i-1}^d / n^d + n^{-1.4} \beta_{i-1}$;
once $\beta_{i-1}^{d-1} / n^{d-1} \leq n^{-0.4}$, we can no longer use the induction.  
Let $i^* \leq \ln \ln n/\ln d + O(1)$ be the point where the induction step no longer applies using
modified random hashing.  At that point
the probability any specific bin with load at least $i^*$ obtains a ball at any time step is at most 
$2n^{-1.4}$.  The probability any bin with load $i^*$ obtains three more balls is thus bounded above
by ${n \choose 3} n (2n^{-1.4})^3 = O(n^{-0.2})$, so the maximum load is $i^* + 3$ with high probability.  
{\hspace*{\fill}\rule{6pt}{6pt}\bigskip}
\end{proof}

\end{document}